\documentclass{sig-alternate-2013}
\newfont{\mycrnotice}{ptmr8t at 7pt}
\newfont{\myconfname}{ptmri8t at 7pt}
\let\crnotice\mycrnotice%
\let\confname\myconfname%

\permission{This is the author's version of the work. It is posted here by permission of ACM for your personal use. Not for redistribution. The definitive version was published in ACM SIGMETRICS Performance Evaluation Review, {42, 1, June 2014}
}
\copyrightetc{Copyright 2014 ACM \the\acmcopyr}
\crdata{\\http://dx.doi.org/10.1145/2591971.2591981}

\newtheorem{theorem}{Theorem}
\newdef{definition}{Definition}
\newtheorem{lemma}{Lemma}
\usepackage[pass,paperwidth=8.5in,paperheight=11in]{geometry}
\usepackage{graphicx}

\begin{document}

\title{Data Dissemination Performance \\in Large-Scale Sensor Networks}
\numberofauthors{2} 
\author{
\alignauthor
Thomas M.M. Meyfroyt, Sem C. Borst, Onno J. Boxma\\
       \affaddr{Eindhoven University of Technology}\\
       \affaddr{P.O. Box 513, 5600 MB Eindhoven, The Netherlands}\\
       \email{\{t.m.m.meyfroyt, s.c.borst, o.j.boxma\}@tue.nl}      
\alignauthor
Dee Denteneer\\
       \affaddr{Philips Research}\\
       \affaddr{HTC 34, 5656 AE Eindhoven, The Netherlands}\\
       \email{dee.denteneer@philips.com}       
    }
\date{\today}

\maketitle

\begin{abstract}
As the use of wireless sensor networks increases, the need for (energy-)efficient and reliable broadcasting algorithms grows. Ideally, a broadcasting algorithm should have the ability to quickly disseminate data, while keeping the number of transmissions low. In this paper we develop a model describing the message count in large-scale wireless sensor networks. We focus our attention on the popular Trickle algorithm, which has been proposed as a suitable communication protocol for code maintenance and propagation in wireless sensor networks. Besides providing a mathematical analysis of the algorithm, we propose a generalized version of Trickle, with an additional parameter defining the length of a listen-only period. This generalization proves to be useful for optimizing the design and usage of the algorithm. For single-cell networks we show how the message count increases with the size of the network and how this depends on the Trickle parameters. Furthermore, we derive distributions of inter-broadcasting times and investigate their asymptotic behavior. Our results prove conjectures made in the literature concerning the effect of a listen-only period. Additionally, we develop an approximation for the expected number of transmissions in multi-cell networks. All results are validated by simulations. 
\end{abstract}


\section{Introduction}
Wireless sensor networks (WSNs) have become more and more popular in the last few years and have many applications \cite{trick1}. These networks consist of compact, inexpensive sensor units that can communicate with each other by wireless transmissions. They require efficient and reliable communication protocols that can quickly propagate new information, while keeping the number of transmissions low, in order to conserve energy and maximize the lifetime of the network. Several data dissemination protocols have been proposed in recent years for this purpose \cite{trick13,trick17, trick3, trick11, trick12}. 

In \cite{trick3} the Trickle algorithm has been proposed in order to effectively and efficiently distribute and maintain code in wireless sensor networks. Trickle relies on a ``polite gossip" policy to quickly propagate updates across a network of nodes. It uses a counter method to reduce the number of redundant transmissions in a network and to prevent a broadcast storm \cite{trick14}. This makes Trickle also a very energy-efficient and popular method of maintaining a sensor network. The algorithm has been standardized by the IETF as the mechanism that regulates the transmission of the control messages used to create the network graph in the IPv6 Routing Protocol for Low power and Lossy Networks (RPL) \cite{trick2}. Additionally, it is used in the Multicast Protocol for Low power and Lossy Networks (MPL), which provides IPv6 multicast forwarding in constrained networks and is currently being standardized \cite{trick20}. Since the algorithm has such a broad applicability, the definition of Trickle has been documented in its own IETF RFC 6206 \cite{trick4}.

Because of the popularity of the algorithm, it is crucial to gain insight in how the various Trickle parameters influence network performance measures, such as energy consumption, available bandwidth and latency. It is clear that such insights are crucial for optimizing the performance of wireless sensor networks. However, there are only a few results in that regard, most of them obtained via simulation studies \cite{trick5, trick3}. 

Some analytical results are obtained in \cite{trick7, trick6, trick5}.
First, in \cite{trick5} qualitative results are provided for the scalability of the Trickle algorithm, but a complete analysis is not given. More specifically, the authors of \cite{trick5} state that in a lossless single-cell network without a listen-only period the expected number of messages per time interval scales as $\mathcal{O}(\sqrt{n})$. Here $n$ is the number of nodes in the network and single-cell means that all the nodes in the network are within communication range from each other. Supposedly, introducing a listen-only period bounds the expected number of transmissions by a constant. The analysis in our paper confirms these claims and provides more explicit results. Secondly, in \cite{trick6} a model is developed for estimating the message count of the Trickle algorithm in a multi-cell network, assuming a uniformly random spatial distribution of the nodes in the network. However, the influence of specific Trickle parameters on the message count is not explicitly addressed. Lastly, in \cite{trick7} an analytical model is developed for the time it takes the Trickle algorithm to update a network, through the use of Laplace transforms.  To the best of the authors' knowledge, no other analytical models for the performance of the Trickle algorithm have been published yet. 

The goal of this paper is to develop and analyze stochastic models describing the Trickle algorithm and gain insight in how the Trickle parameters influence inter-transmission times and the message count. These insights could help optimize the energy-efficiency of the algorithm and consequently the lifetime of wireless sensor networks. Furthermore, knowing how the inter-transmission times depend on the various parameters, could help prevent hidden-node problems in a network and optimize the capacity of wireless sensor networks.  Additionally, our models are relevant for the analysis of other communication protocols that build upon Trickle, such as CTP \cite{trick19}, Deluge \cite{trick15} and Melete \cite{trick12}, and could give insight in their performance.

As key contributions of this paper, we first propose a generalized version of the Trickle algorithm by introducing a new parameter $\eta$, defining the length of a listen-only period. This addition proves to be useful for optimizing the design and usage of the algorithm. Furthermore, we derive the distribution of an inter-transmission time and the joint distribution of consecutive inter-transmission times for large-scale single-cell networks. We show how they depend on the Trickle parameters and investigate their asymptotic behavior.  Additionally, we show that in a single-cell network without a listen-only period the expected number of transmissions per time interval is unbounded and grows as $\sqrt{2n} \Gamma\left[\frac{k+1}{2}\right]/\Gamma\left[\frac{k}{2}\right]$, where $k$ is called the redundancy constant and $n$ the number of nodes as before. When a listen-only period of $\eta>0$ is introduced, the message count is bounded by $k/\eta$ from above. We then use the results from the single-cell analysis to develop an approximation for the transmission count in multi-cell networks. All our results are compared and validated with simulation results and prove to be accurate, even for networks consisting of relatively few nodes. 
\subsection{Organization of the Paper}
The remainder of this paper is organized as follows. In Section 2 we give a detailed description of the Trickle algorithm. Furthermore, we introduce a new parameter defining the length of a listen-only period and discuss its relevance. In Section 3 we develop a mathematical model describing the behavior of the Trickle algorithm in large-scale single-cell networks. We first briefly list the main results, before presenting the details of the model and its analysis. This is followed by a brief discussion of our results and finally we validate our findings with simulations. In Section 4, we use the results from Section 3 to approximate the message count in multi-cell networks and again compare our approximations with simulation results. In Section 5 we make some concluding remarks.
\vfill\eject
\section{The Trickle Algorithm}
The Trickle algorithm has two main goals. First, whenever a new update enters the network, it must be propagated quickly throughout the network. Secondly, when there is no new update in the network, communication overhead has to be kept to a minimum. 

The Trickle algorithm achieves this by using a ``polite gossip" policy. Nodes divide time into intervals of varying length. During each interval a node will broadcast its current information, if it has heard fewer than, say, $k$ other nodes transmit the same information during that interval, in order to check if its information is up to date. If it has recently heard at least $k$ other nodes transmit the same information it currently has, it will stay quiet, assuming there is no new information to be received. Additionally, it will increase the length of its intervals, decreasing its broadcasting rate. Whenever a node receives an update or hears outdated information, it will reduce the length of its intervals, increasing its broadcasting rate, in order to quickly update nodes that have outdated information. This way inconsistencies are detected and resolved quickly, while keeping the number of transmissions low. 

\subsection{Algorithm Description}
We now describe the Trickle algorithm in its most general form (see also \cite{trick3}). The algorithm has four parameters:
 \begin{itemize}
 \item A threshold value $k$, called the redundancy constant.
 \item The maximum interval length $\tau_h$.
 \item The minimum interval length $\tau_l$.
 \item The listen-only parameter $\eta$, defining the length of a listen-only period.
 \end{itemize}
Furthermore, each node in the network has its own timer and keeps track of three variables:
  \begin{itemize}
 \item The current interval length $\tau$.
 \item A counter $c$, counting the number of messages heard during an interval.
 \item A broadcasting time $\theta$ during the current interval.
 \end{itemize}
The behavior of each node is described by the following set of rules:
\begin{enumerate}
 \item At the start of a new interval a node resets its timer and counter $c$ and sets $\theta$ to a value in $[\eta\tau,\tau]$ uniformly at random.
 \item When a node hears a message that is consistent with the information it has, it increments $c$ by 1.
 \item When a node's timer hits time $\theta$, the node broadcasts its message if $c<k$.
 \item When a node's timer hits time $\tau$, it doubles its interval length $\tau$ up to $\tau_h$ and starts a new interval.
 \item When a node hears a message that is inconsistent with its own information, then if $\tau>\tau_l$ it sets $\tau$ to $\tau_l$ and starts a new interval, otherwise it does nothing.
\end{enumerate}
\begin{figure}[!h]
\centering
\includegraphics[width=8.2cm]{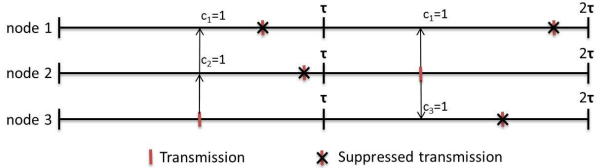}
\caption{Example of three synchronized nodes
using the Trickle algorithm.}
\label{fig:1}
\end{figure}
In Figure \ref{fig:1} an example is depicted of a network consisting of three nodes using the Trickle algorithm with $k=1$ and $\tau=\tau_h$ for all nodes. During the first interval, node 3 is the first node that attempts to broadcast and consequently it is successful. The broadcasts of nodes 1 and 2 during that interval are then suppressed. During the second interval the broadcast of node 2 suppresses the other broadcasts.

Note that in the example in Figure \ref{fig:1}, the intervals of the three nodes are synchronized. However, in general, the times at which nodes start their intervals need not be synchronized. In a synchronized network, all the nodes start their intervals at the same time, while in an unsynchronized network this is not necessarily the case. In practice, networks will generally not be synchronized, since synchronization requires additional communication and consequently imposes energy overhead. Furthermore, as nodes get updated and start new intervals, they automatically lose synchronicity. 

\subsection{The Listen-Only Parameter {\large$ \eta$}}
Note that the parameter $\eta$ is not introduced in the description of the Trickle algorithm in \cite{trick4} and \cite{trick3}. We have added this parameter ourselves, so we can analyze a more general version of the Trickle algorithm. The authors of \cite{trick3} propose to always use a listen-only period of half an interval, i.e., $\eta=\frac{1}{2}$, because of the so-called short-listen problem, which is discussed in the same paper. When no listen-only period is used, i.e., $\eta=0$, sometimes nodes will broadcast soon after the beginning of their interval, listening for only a short time, before anyone else has a chance to speak up. If we have a perfectly synchronized network this does not give a problem, because the first $k$ transmissions will simply suppress all the other broadcasts during that interval. However in an unsynchronized network, if a node has a short listening period, it might broadcast just before another node starts its interval and that node possibly also has a short listening period. This possibly leads to a lot of redundant messages and is referred to as the short-listen problem.

In \cite{trick3} it is claimed that not having a listen-only period and $k=1$ makes the number of messages per time interval scale as $\mathcal{O}(\sqrt{n})$, due to the short-listen problem. When a listen-only period of $\tau/2$ is used, the expected number of messages per interval is supposedly bounded by 2, resolving the short-listen problem and improving scalability.

However, introducing a listen-only period also has its disadvantages. Firstly, when a listen-only period of $\tau/2$ is used, newly updated nodes will always have to wait for a period of at least $\tau_l/2$, before attempting to propagate the received update. Consequently, in an $m$-hop network, the end-to-end delay is at least $m\frac{\tau_l}{2}$. Hence, a listen-only period greatly affects the speed at which the Trickle algorithm can propagate updates. 
\begin{figure}[!h]
\centering
\includegraphics[width=8cm]{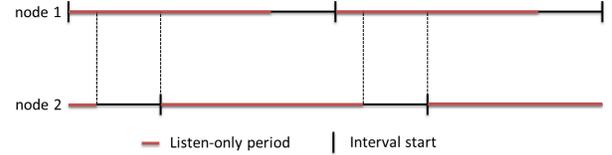}
\caption{One node carries the complete transmission load of the network.}
\label{fig:2}
\end{figure}
Secondly, introducing a listen-only period has a negative effect on the load distribution. This is illustrated by Figure \ref{fig:2}, where node 2's broadcasting intervals completely overlap with node 1's listen-only periods and vice versa. Consequently, one node will always transmit and suppress the other node's transmissions, depending on which node first starts broadcasting. The probability of having an uneven load distribution increases as $\eta$ increases.

For these reasons one might want to consider using a shorter listen-only period, which raises the question what length is optimal. Therefore we have added the parameter $\eta$, which allows us to investigate the effect of using a listen-only period of general length on the message count.

\section{Single-Cell Network}
In this section we will develop a mathematical model for the message count in single-cell networks. We consider a steady-state regime, where all nodes are up to date and their interval lengths have settled to $\tau=\tau_h$. Since in this state all communication is essentially redundant, energy-consumption is the critical performance measure.  We denote by $N^{(k,n)}$ the number of transmissions during an interval of length $\tau_h$ for a given threshold value $k$ in a single-cell network consisting of $n$ nodes.  Without loss of generality, we will assume that $\tau_h=1$. We are interested in determining the behavior of $\mathbb{E}[N^{(k,n)}]$ as $n$ grows large. Additionally, let us denote an inter-transmission time for a given value of $k$ and cell-size $n$ by $T^{(k,n)}$. We will determine the distribution of $T^{(k,n)}$ as $n$ grows large.  Whenever we write $a(n)\sim b(n)+c$, where $c$ is a constant, this means $a(n)-c$ is approximately equal to $b(n)$ as $n$ grows large and more formally 
\[a(n)\sim b(n)+c \text{ denotes } \lim_{n\rightarrow \infty} (a(n)-c)/b(n)=1.\] 
We first give a brief overview of our main results, before presenting the analytical model and its analysis. We then discuss our results and show through simulations, that our asymptotic results also provide good estimates for networks consisting of relatively few nodes. 

\subsection{Main Results}
First we look at the case $k=1$. We show that the cumulative distribution function of $T^{(1,n)}$ for large $n$ behaves as
\begin{equation}
F^{(1,n)}(t)=
\begin{cases}
  0, & t<\eta,\\
  1-e^{-\frac{n}{2}\frac{(t-\eta)^2}{1-\eta}} , & t\geq \eta.
  \end{cases} 
\end{equation}
This lets us deduce that 
\begin{equation} 
\mathbb{E}\left[T^{(1,n)}\right]\sim\eta+\sqrt{\frac{\pi(1-\eta)}{2n}}.
\end{equation}
We conclude that $\mathbb{E}[N^{(1,n)}]\sim\sqrt{\frac{2n}{\pi}}$, when $\eta=0$. This proves the claim from \cite{trick3}, that when no listen-only period is used, $\mathbb{E}[N^{(1,n)}]=\mathcal{O}(\sqrt{n})$, and shows that the pre-factor is $\sqrt{\frac{2}{\pi}}$. Furthermore, when $\eta>0$, $\mathbb{E}[N^{(1,n)}]\uparrow \frac{1}{\eta}$, with a convergence rate of $\sqrt{n}$, proving another claim from \cite{trick3}.

For the case $k\geq 2$, we first derive the density function for the distribution of $k-1$ consecutive inter-broadcasting times. We then use this result to deduce that the density function of the distribution of $T^{(k,n)}$ for large $n$ behaves as
\begin{equation}
\resizebox{.88\linewidth}{!}{$\displaystyle
f^{(k,n)}(t)=\frac{C_{(k,n)}}{(k-2)!}\int_{0}^\infty\lambda(t\text{ }\vert\text{ }\nu)\nu^{k-2}\exp\left[-\frac{n(t+\nu-\eta)^2}{2(1-\eta)}\right]\text{d}\nu.
$}
\end{equation} 
Here
\begin{equation}
\lambda(t\text{ }\vert\text{ }\nu)=
\left\{\begin{array}{ll}
0, & \hbox{$t+\nu<\eta,$}\\\\
\frac{n}{1-\eta}(t+\nu-\eta), & \hbox{$t+\nu\geq\eta$,} \end{array} \right.
\end{equation}
and 
\begin{equation}
\resizebox{.88\linewidth}{!}{$\displaystyle
C_{(k,n)}=\left(\frac{\eta^{k-1}}{(k-1)!}+\int_\eta^\infty \frac{t^{k-2}}{(k-2)!}\exp\left[-\frac{n (t-\eta)^2}{2(1-\eta)}\right]\text{d}t\right)^{-1}
$}.
\end{equation}
Additionally we find for the $j$th moment of $T^{(k,n)}$:
\begin{equation}
\mathbb{E}\left[\left(T^{(k,n)}\right)^j\right]\sim j!\frac{C_{(k,n)}}{C_{(k+j,n)}}.
\end{equation}
Hence, for $\eta=0$, $\mathbb{E}[N^{(k,n)}]\sim\sqrt{2n}\Gamma\left[\frac{k+1}{2}\right]/\Gamma\left[\frac{k}{2}\right]$, which is again $\mathcal{O}(\sqrt{n})$. Moreover, when $\eta>0$, $\mathbb{E}[N^{(k,n)}]\uparrow \frac{k}{\eta}$, as $n \rightarrow \infty$, with a convergence rate of $\sqrt{n}$.

We then use these results to derive the asymptotic distributions of inter-broadcasting times. When $\eta>0$ and $k\geq2$, we show that
\begin{equation}
\frac{1}{\eta}T^{(k,n)}\xrightarrow{d}\text{Beta}(1,k-1)\text{, as }n\rightarrow \infty,
\end{equation}
and
\begin{equation}
\frac{k}{\eta}T^{(k,n)}\xrightarrow{d}\text{Exp}(1)\text{, as }n\rightarrow \infty\text{ and }k\rightarrow \infty.
\end{equation}

For the case $\eta=0$ we show that the density function $f^{(k)}(t)$ of the limiting distribution of $\sqrt{\frac{n}{2}}T^{(k,n)}$ as $n\rightarrow \infty$, satisfies
\begin{equation}
f^{(k)}(t)=\frac{k-2}{k-3}f^{(k-2)}(t)-\frac{\Gamma\left[\frac{k}{2}\right]}{\Gamma\left[\frac{k-1}{2}\right]}\frac{2t}{k-3}f^{(k-1)}(t),
\end{equation}
where
\begin{align*}
f^{(2)}(t)&=\frac{2}{\sqrt{\pi}}e^{-t^2},\\
f^{(3)}(t)&=\sqrt{\pi}\text{erfc}(t).
\end{align*}
Lastly we show that
\begin{equation}
\sqrt{nk}T^{(k,n)}\xrightarrow{d}\text{Exp}(1)\text{, as }n\rightarrow \infty\text{ and }k\rightarrow \infty.
\end{equation}

\subsection{Analytical Model}
Suppose we have a single-cell network consisting of $n$ nodes which are all perfect receivers and transmitters. Furthermore, we assume that all nodes are up to date and $\tau=\tau_h=1$ for all nodes. Lastly, we assume that the interval skew of the nodes is uniformly distributed, meaning each node has one interval starting at some time in the interval $[0,1)$ uniformly at random.

An important first observation we can make under these assumptions, is that the properly scaled process of nodes attempting to broadcast behaves as a Poisson process with rate $1$ as $n$ grows large. 
\begin{lemma}
Let $N_n$ be the point process of times that nodes attempt to broadcast in a single cell consisting of $n$ nodes with $\eta \in[0,1]$. Then if we dilate the  timescale by a factor $n$, the process $N_n$ converges weakly to a Poisson process with rate 1 as $n$ grows large.
\end{lemma}
\begin{proof}
See Appendix A.
\end{proof}
We now use Lemma 1 to analyze the message count. In light of Lemma 1, we treat the process of nodes attempting to broadcast as a Poisson process with rate $n$. We will determine the PDF and CDF of an inter-transmission time $T_P^{(k,n)}$ for this process, which we denote by $f^{(k,n)}(t)$ and $F^{(k,n)}(t)$ respectively. Recall, that by $T^{(k,n)}$ we denote an inter-transmission time for a given value of $k$ and cell-size $n$, for the original broadcasting process, which is not Poisson. We will use the fact that the distribution functions of $T_P^{(k,n)}$  and $T^{(k,n)}$ will behave the same for large $n$ (because of Lemma 1), to investigate the behavior of $T^{(k,n)}$. It is important to keep in mind that we are analyzing the Poisson broadcasting process, and all results for $T^{(k,n)}$ (and $N^{(k,n)}$) are only true asymptotically, for $n\rightarrow \infty$.
\subsubsection{{\large$k=1$}}
We first consider the case $k=1$. Suppose at time $0$ a broadcast occurs. Now assume at time $t$ a node ends its listening period and attempts to broadcast. In order for the node's transmission to not be suppressed, it must have started its listening period after time 0. If it started listening before this time, it would have heard the transmission at time $0$ and its broadcast would have been suppressed. This means that the node's broadcast will be successful if its corresponding timer was smaller than $t$. Since broadcasting times are picked uniformly in $[\eta,1]$, this probability is 0 if $t<\eta$ and $\frac{t-\eta}{1-\eta}$ otherwise. Hence, if we define $\lambda(t)$ to be the instantaneous rate at which successful broadcasts occur, we can write \[ 
\lambda(t)=
\left\{\begin{array}{ll}
0, & \hbox{$t<\eta,$}\\\\
\frac{n}{1-\eta}(t-\eta), & \hbox{$t \geq\eta$.} \end{array} \right.
\]
It is well known that the hazard rate $\lambda(t)=\frac{f^{(1,n)}(t)}{1-F^{(1,n)}(t)}$ uniquely determines $F^{(1,n)}(t)=\mathbb{P}[T_P^{(1,n)}\leq t]$ (see \cite{trick9}, Theorem 2.1):
\begin{equation}\label{densk1}
\resizebox{.88\linewidth}{!}{$
\displaystyle
F^{(1,n)}(t)=1-e^{-\int_0^t\lambda(u)\text{ d}u}=\left\{
\begin{array}{ll}
  0, & \hbox{$t<\eta$,}\\\\
  1-e^{-\frac{n}{2}\frac{(t-\eta)^2}{1-\eta}} , & \hbox{$t\geq \eta$.}\end{array} \right.
  $}
\end{equation}
Hence, $\sqrt{\frac{n}{1-\eta}}(T_P^{(1,n)}-\eta)$ is a Rayleigh distributed random variable with scale parameter $\sigma=1$. Therefore,
\[\mathbb{E}\left[T^{(1,n)}\right]\sim\mathbb{E}\left[T_P^{(1,n)}\right]=\eta+\sqrt{\frac{\pi(1-\eta)}{2n}}.\]
We conclude
\begin{equation}\label{exp1}
\resizebox{.88\linewidth}{!}{$
\displaystyle
\mathbb{E}\left[N^{(1,n)}\right]=\left(\mathbb{E}\left[T^{(1,n)}\right]\right)^{-1}\sim\left(\sqrt{\frac{\pi(1-\eta)}{2n}}+\eta\right)^{-1}.
$}
\end{equation}
Hence, for the case $\eta=0$ we find $\mathbb{E}[N^{(1,n)}]\sim\sqrt{\frac{2n}{\pi}}$. This proves the claim from \cite{trick3}, that $\mathbb{E}[N^{(1,n)}]=\mathcal{O}(\sqrt{n})$ when no listen-only period is used. 
For $\eta>0$ we get from (\ref{exp1}) that
\[\mathbb{E}\left[N^{(1,n)}\right]\sim \frac{1}{\eta}-\frac{1}{\eta^2}\sqrt{\frac{\pi(1-\eta)}{2n}}+\mathcal{O}(n^{-1}).\]
This implies $\mathbb{E}[N^{(1,n)}]\uparrow \frac{1}{\eta}$ from below, proving the claim that introducing a listen-only period bounds the number of transmissions per interval by a constant.
\subsubsection{{\large$k=2$}}
Let us now look at the case $k=2$. We can apply a similar reasoning. Suppose again that at time $0$ a broadcast occurs and let $-T_{-1}$  be the time of the last broadcast before time $0$. Now similarly as before, a broadcasting attempt at time $t$ will be successful, if the corresponding node started its listening interval after time 
$-T_{-1}$. Hence, the instantaneous rate at which successful broadcasts occur conditioned on $T_{-1}=\nu$ is given by 
\begin{equation}\label{lambda}
\lambda(t\text{ }\vert\text{ }\nu)=
\left\{\begin{array}{ll}
0, & \hbox{$t+\nu<\eta,$}\\\\
\frac{n}{1-\eta}(t+\nu-\eta), & \hbox{$t+\nu\geq\eta$.} \end{array} \right.
\end{equation}
Hence,
\begin{equation}\label{trans}
\resizebox{.88\linewidth}{!}{$
\displaystyle
\begin{array}{c}
\mathbb{P}\left[T_P^{(2,n)}\leq t\text{ }\big\vert\text{ }T_{-1}=\nu\right]=1-\exp\left[-\int_0^t \lambda(u\text{ }\vert\text{ }\nu)\text{d}u\right]\\\\
=\begin{cases}
 0, & t+\nu<\eta, \\
 1-\exp\left[-\frac{n(t+\nu-\eta)^2}{2(1-\eta)}\right], & \nu<\eta\text{ }\wedge\text{ }t+\nu\geq\eta, \\
1-\exp\left[-\frac{n(t^2/2+t(\nu-\eta))}{1-\eta}\right], & \nu\geq\eta.  \end{cases}
\end{array}
$}
\end{equation}
This tells us that for $k=2$ the sequence of consecutive inter-transmission times $\boldsymbol{T_P^{(2,n)}}=\{T_{P,i}^{(2,n)}\}_{i=0}^\infty$ forms a Markov chain with transition probabilities as in (\ref{trans}). Clearly the chain is $\phi$-irreducible, that is every set $A\subset \mathbb{R}^+$ with Lebesgue measure $\mu_L(A)>0$ will eventually be reached, since we can reach any state within two steps. Furthermore, the chain is Harris recurrent since every set $A\subset\mathbb{R}^+$ with $\mu_L(A)>0$ will eventually be reached. Then, since the chain is also aperiodic, the following theorem applies (see \cite{trick10}, Theorem 13.0.1):
\begin{theorem}
Suppose $\boldsymbol{X}$ is a Markov chain. Suppose it is $\phi$-irreducible, Harris recurrent and aperiodic and it admits a finite invariant probability measure $f(t)$, then this measure is unique.
\end{theorem}
Hence, if our chain $\boldsymbol{T_P^{(2,n)}}$ admits an invariant distribution $f^{(2,n)}(t)$, it is the unique stationary distribution and it should satisfy
\begin{equation}\label{eq1}
F^{(2,n)}(t)= \int_0^\infty \mathbb{P}\left[T_P^{(2,n)}\leq t\text{ }\big\vert\text{ }T_{-1}=\nu\right]\text{ d}F^{(2,n)}(\nu).
\end{equation}
Substituting Equation (\ref{trans}) into (\ref{eq1}) and differentiating with respect to $t$ gives an integral equation for $f^{(2,n)}(t)$:
\begin{equation}\label{inteq2}
\resizebox{.88\linewidth}{!}{$\displaystyle
f^{(2,n)}(t)=\int_0^\infty f^{(2,n)}(\nu)\lambda(t\text{ }\vert\text{ }\nu)\exp\left[-\int_0^t\lambda(u\text{ }\vert\text{ }\nu)\text{d}u\right]\text{d}\nu.
$
}
\end{equation}
By substituting Equation (\ref{lambda}) we can write the right-hand side of Equation (\ref{inteq2}) for $t<\eta$ as
\begin{equation}\label{inteq3}
\resizebox{.88\linewidth}{!}{
$
\begin{array}{ll}
 &\displaystyle\int_{\eta-t}^\eta f^{(2,n)}(\nu)  \frac{n(t+\nu-\eta)}{1-\eta}\exp\left[-\frac{n(t+\nu-\eta)^2}{2(1-\eta)}\right]\text{d}\nu\vspace{.2cm}\\
+&\displaystyle\int_{\eta}^\infty f^{(2,n)}(\nu)  \frac{n(t+\nu-\eta)}{1-\eta}\exp\left[-\frac{n(t^2/2+t(\nu-\eta))}{1-\eta}\right]\text{d}\nu,
\end{array}
$
}
\end{equation}
and for $t \geq\eta$ we get
\begin{equation}\label{inteq4}
\resizebox{.88\linewidth}{!}{$
\begin{array}{ll}
&\displaystyle\int_{0}^\eta f^{(2,n)}(\nu)  \frac{n(t+\nu-\eta)}{1-\eta}\exp\left[-\frac{n(t+\nu-\eta)^2}{2(1-\eta)}\right]\text{d}\nu\vspace{.2cm}\\
+&\displaystyle\int_{\eta}^\infty f^{(2,n)}(\nu)  \frac{n(t+\nu-\eta)}{1-\eta}\exp\left[-\frac{n(t^2/2+t(\nu-\eta))}{1-\eta}\right]\text{d}\nu.
\end{array}
$}
\end{equation}
We show that the solution to this equation is given by
\begin{equation}\label{sol1}
f^{(2,n)}(t)=\begin{cases} C_{(2,n)} ,& t<\eta,\\
C_{(2,n)} \exp\left[-\frac{n(t-\eta)^2}{2(1-\eta)}\right],& t\geq\eta,
\end{cases}
\end{equation}
where $C_{(2,n)}=\left(\eta+\sqrt{\frac{(1-\eta)\pi}{2n}}\right)^{-1}$ is a normalization constant. 
Substitution of (\ref{sol1}) reduces (\ref{inteq3}) for $t<\eta$ to
\begin{gather*}
\displaystyle
C_{(2,n)} \int_{\eta-t}^\infty   \frac{n(t+\nu-\eta)}{1-\eta}\exp\left[-\frac{n(t+\nu-\eta)^2}{2(1-\eta)}\right]\text{d}\nu
\\ 
=C_{(2,n)} \int_{0}^\infty  z \exp\left[-\frac{z^2}{2}\right]\text{d}z = C_{(2,n)}.
\end{gather*}
Similarly, for $t\geq\eta$ substitution of (\ref{sol1}) reduces Equation (\ref{inteq4}) to
\begin{gather*} \displaystyle
C_{(2,n)} \int_{0}^\infty   \frac{n(t+\nu-\eta)}{1-\eta}\exp\left[-\frac{n(t+\nu-\eta)^2}{2(1-\eta)}\right]\text{d}\nu\\
\resizebox{1\linewidth}{!}{$\displaystyle
=C_{(2,n)} \int_{\frac{\sqrt{n}(t-\eta)}{\sqrt{1-\eta}}}^\infty  z \exp\left[-\frac{z^2}{2}\right]\text{d}z = C_{(2,n)} \exp\left[-\frac{n(t-\eta)^2}{2(1-\eta)}\right].
$}
\end{gather*}
Thus, we have shown that $f^{(2,n)}(t)$ as given in Equation (\ref{sol1}) is the unique solution to (\ref{inteq2}). Calculating the expectation of $T_P^{(2,n)}$ we conclude
\begin{equation}\label{exp2}
\mathbb{E}\left[T^{(2,n)}\right]\sim\mathbb{E}\left[T_P^{(2,n)}\right]=\frac{\left(\frac{\eta^2}{2}+\eta\sqrt{\frac{\pi(1-\eta)}{2n}}+\frac{1-\eta}{n}\right)}{\left(\eta+\sqrt{\frac{\pi(1-\eta)}{2n}}\right)}.
\end{equation}
For $\eta=0$ this gives $\mathbb{E}[N^{(2,n)}]\sim\sqrt{\frac{\pi n}{2}}$, which is again $\mathcal{O}(\sqrt{n})$. If a listen-only period is used, i.e., $\eta>0$, then from (\ref{exp2}) we get
\[\mathbb{E}\left[N^{(2,n)}\right]\sim \frac{2}{\eta}-\frac{1}{\eta^2}\sqrt{\frac{2\pi(1-\eta)}{n}}+\mathcal{O}(n^{-1}),\]
implying $\mathbb{E}[N^{(2,n)}]\uparrow \frac{2}{\eta}$ from below as $n$ grows large.
\subsubsection{General {\large$k\geq2$}}
We move on to the more general case $k\geq 2$. Again suppose that at time $0$ a broadcast occurs and let $-T_{-(k-1)}$  be the time of the $(k-1)$th broadcast before time $0$. Now similarly as before, a broadcasting attempt at time $t$ will be successful, if the corresponding node started its listening interval after time 
$-T_{-(k-1)}$. Hence, the instantaneous rate at which broadcasts occur conditioned on $T_{-(k-1)}=\nu$ is again given by Equation (\ref{lambda}).
Thus, the sequence $\boldsymbol{T_P^{(k,n)}}=\left\{\left(T^{(k,n)}_{P,i},\text{ ..., }T^{(k,n)}_{P,i+k-2}\right)\right\}_{i=0}^\infty$ forms a Markov chain. Since this chain is able to move from any state to any other state within $k$ steps, it is easily seen that again Theorem 1 applies. Consequently, if it admits an invariant probability measure, it is unique. Let us denote the steady-state joint probability density function of $k-1$ consecutive inter-transmission times by $\tilde{f}^{(k,n)}(t_1,\text{ ..., }t_{k-1})$. Similar to (\ref{inteq2}), $\tilde{f}^{(k,n)}(t_1,\text{ ..., }t_{k-1})$ can then be written as
\begin{equation}\label{inteq5}
\resizebox{.88\linewidth}{!}{$\displaystyle
\begin{split}\tilde{f}^{(k,n)}(t_1,\text{..., }t_{k-1})=\int_0^\infty \tilde{f}^{(k,n)}(t_2, &\text{ ..., } t_k)\lambda\left(t_1\text{ }\big\vert\text{ }\sum\limits_{i=2}^{k}t_i\right)&\\\
&\cdot\exp\left[-\int_0^{t_1}\lambda\left(u\text{ }\big\vert\text{ }\sum\limits_{i=2}^{k}t_i\right)\text{d}u\right]\text{d}t_k.\end{split}$}
\end{equation}
Consequently, $\tilde{f}^{(k,n)}(t_1,\text{ ..., }t_{k-1})$ has the same form as the solution in (\ref{sol1}) and is given by
\begin{equation}\label{sol2}
\resizebox{.88\linewidth}{!}{$\displaystyle
\tilde{f}^{(k,n)}(t_1, \text{ ..., }t_{k-1})=\begin{cases} C_{(k,n)} ,& \sum\limits_{i=1}^{k-1}t_i<\eta,\\
C_{(k,n)} \exp\left[-\frac{n}{2(1-\eta)}\left(\sum\limits_{i=1}^{k-1}t_i-\eta\right)^2\right],& \sum\limits_{i=1}^{k-1}t_i\geq\eta.
\end{cases}
$}
\end{equation}
Here the normalization constant $C_{(k,n)}$ satisfies
\begin{equation}\label{norm}
\resizebox{.88\linewidth}{!}{$\displaystyle
C_{(k,n)}=\left(\frac{\eta^{k-1}}{(k-1)!}+\int_\eta^\infty \frac{t^{k-2}}{(k-2)!}\exp\left[-\frac{n (t-\eta)^2}{2(1-\eta)}\right]\text{d}t\right)^{-1},
$}
\end{equation}
where we have used the following lemma:
\begin{lemma}
Let $k\in\mathbb{N}$ and $F(x)$ be a positive real-valued integrable function. Assume that $\int_0^\infty x^k F(x)\text{d}x < \infty $, then
\[ \int_0^\infty\cdots\int_0^\infty F\left(\sum_{i=1}^{k+1}x_i\right)\text{d}x_1\cdots\text{d}x_{k+1}=\int_0^\infty\frac{x^k}{k!}F(x)\text{d}x.\]
\end{lemma}
\begin{proof}
By writing $x=\sum_{i=1}^{k+1}x_i$ and a change of variables, we can write
\[\begin{array}{ll}&\displaystyle \int_0^\infty\cdots\int_0^\infty F\left(\sum_{i=1}^{k+1}x_i\right)\text{d}x_1\cdots\text{d}x_{k+1}\vspace{0.2cm}\\
\displaystyle=&\displaystyle\int_0^\infty F(x) \underset{\resizebox{.25\linewidth}{!}{$\begin{array}{c} x_1+\cdots+x_k\leq x\\ x_i\geq 0\text{ for }1\leq i \leq k\end{array}$}}{\int \cdots \int} \text{d}x_1\cdots\text{d}x_{k}\text{ d}x.\\\end{array}
\]
The inner $k$-tuple integral is equal to the volume of the $k$-dimensional simplex $\{(x_1,\cdots,x_k)\text{ }\vert\text{ }\sum_{i=1}^{k}x_i\leq x\text{ and }x_i\geq 0\text{ for }1\leq i \leq k\}$, which is known to be $x^k/k!$ \cite{trick18}. Applying this result completes the proof.
\end{proof}
Alternatively, by a change of variables and splitting the integral, we can write (\ref{norm}) in terms of a finite sum as
\begin{equation}\nonumber
\resizebox{1\linewidth}{!}{$\displaystyle
\displaystyle C_{(k,n)}=\left(\frac{\eta^{k-1}}{(k-1)!}+\frac{1}{2(k-2)!}\sum_{i=0}^{k-2} \binom{k-2}{i} \eta^{k-i-2}\left(\frac{2(1-\eta)}{n}\right)^{\frac{i+1}{2}}\Gamma\left[\frac{i+1}{2}\right]\right)^{-1},
$}
\end{equation}
which more clearly reveals the role of some of the parameters. 

There are two important observations we can make regarding the normalization constant $C_{(k,n)}$. First, for $\eta=0$ (\ref{norm}) reduces to $C_{(k,n)}=\pi^{-\frac{1}{2}}(2n)^{\frac{k-1}{2}}\Gamma[k/2]$. Secondly, for $\eta>0$ we have that $C_{(k,n)} \downarrow \frac{(k-1)!}{\eta^{k-1}}$ as $n$ grows large.

Let us look at the sequence $\boldsymbol{\Sigma^{(k,n)}}=\{\sum_{j=0}^{k-2}T^{(k,n)}_{P,i+j}\}_{i=0}^\infty$, i.e., the sequence of the sum of $k-1$ consecutive inter-transmissions times. Integrating Equation (\ref{sol2}) and applying Lemma 2, we can easily determine its steady-state probability density function:
\begin{equation}\label{sol3}
\resizebox{.88\linewidth}{!}{$\displaystyle
f_\Sigma^{(k,n)}(s)=\begin{cases} \frac{C_{(k,n)}}{(k-2)!}s^{k-2} ,& s<\eta,\\
\frac{C_{(k,n)}}{(k-2)!} s^{k-2} \exp\left[-\frac{n(s-\eta)^2}{2(1-\eta)}\right],& s\geq\eta.
\end{cases}
$}
\end{equation}
This lets us deduce that for large $n$:
\begin{equation}\label{expec}
\mathbb{E}\left[T^{(k,n)}\right]\sim\frac{1}{k-1}\int_0^\infty s f_\Sigma^{(k,n)}(s)\text{d}s=\frac{C_{(k,n)}}{C_{(k+1,n)}}.
\end{equation}
For the case $\eta=0$ this implies
\begin{equation}\label{expec1}
\mathbb{E}\left[N^{(k,n)}\right]=\left(\mathbb{E}\left[T^{(k,n}\right]\right)^{-1}\sim \sqrt{2n}\frac{\Gamma\left[\frac{k+1}{2}\right]}{\Gamma\left[\frac{k}{2}\right]},
\end{equation}
which is again $\mathcal{O}(\sqrt{n})$.
When a listen-only period is used, i.e., $\eta>0$, 
\begin{equation}\label{expec2}
\mathbb{E}\left[N^{(k,n)}\right]\sim \frac{k}{\eta}-\frac{k}{\eta^2}\sqrt{\frac{\pi(1-\eta)}{2n}}+\mathcal{O}(n^{-1}),
\end{equation}
implying that $\mathbb{E}[N^{(k,n)}]\uparrow \frac{k}{\eta}$ from below as $n\rightarrow \infty$.

Moreover, the function $f_\Sigma^{(k,n)}(s)$ allows us to determine the steady-state probability density function $f^{(k,n)}(t)$. Analogously to Equation (\ref{eq1}), we can write
\begin{equation}\label{eq2}
\resizebox{.88\linewidth}{!}{$\begin{array}{rcl}\displaystyle
\displaystyle F^{(k,n)}(t)&=&\displaystyle\int_0^\infty \mathbb{P}\left[T_P^{(k,n)}\leq t\text{ }\big\vert\text{ }T_{-(k-1)}=\nu\right]\text{ d}\mathbb{P}[T_{-(k-1)}\leq\nu]\\\\\displaystyle
&=&\displaystyle\int_0^\infty \left(1-\exp\left[-\int_0^t \lambda(u\text{ }\vert\text{ }\nu)\text{ d}u\right]\right)f_\Sigma^{(k,n)}(\nu)\text{ d}\nu.
\end{array}
$}
\end{equation}
Differentiating (\ref{eq2}) with respect to $t$ and substituting Equation (\ref{sol3}) gives
\begin{equation}\label{dens}
\resizebox{.88\linewidth}{!}{$\displaystyle
f^{(k,n)}(t)=\frac{C_{(k,n)}}{(k-2)!}\int_{0}^\infty\lambda(t\text{ }\vert\text{ }\nu)\nu^{k-2}\exp\left[-\frac{n(t+\nu-\eta)^2}{2(1-\eta)}\right]\text{d}\nu.
$}
\end{equation}
Substitution of Equation (\ref{lambda}) reduces the right-hand side of (\ref{dens}) for $t<\eta$ to
\begin{equation}\label{dens2}
\resizebox{.88\linewidth}{!}{$\displaystyle
\frac{n}{1-\eta}\frac{C_{(k,n)}}{(k-2)!}\int_{\eta-t}^\infty(t+\nu-\eta)\nu^{k-2}\exp\left[-\frac{n(t+\nu-\eta)^2}{2(1-\eta)}\right]\text{d}\nu,
$}
\end{equation}
and for $t\geq \eta$ the right-hand side reduces to
\begin{equation}\label{dens3}
\resizebox{.88\linewidth}{!}{$\displaystyle
\frac{n}{1-\eta}\frac{C_{(k,n)}}{(k-2)!}\int_{0}^\infty(t+\nu-\eta)\nu^{k-2}\exp\left[-\frac{n(t+\nu-\eta)^2}{2(1-\eta)}\right]\text{d}\nu.
$}
\end{equation}

Lastly, Equation (\ref{eq2}) in combination with Equations (\ref{lambda}) and (\ref{sol3}) allows to calculate the moments of $T_P^{(k,n)}$as follows:
\begin{equation}\label{mom0}
\resizebox{.88\linewidth}{!}{$\begin{array}{c}\displaystyle
\mathbb{E}\left[\left(T_P^{(k,n)}\right)^j\right]= j\int_0^\infty t^{j-1} (1-F^{(k,n)}(t))\text{ d}t=\\\\\displaystyle
\frac{j C_{(k,n)}}{(k-2)!} \int_0^\infty \int_0^\infty t^{j-1}v^{k-2}\exp\left[-\frac{n\max[(t+v-\eta),0]^2}{2(1-\eta)}\right]\text{ d}t\text{d}\nu.
\end{array}
$}
\end{equation}
Here we have used the $\max$ function, so we do not have to split the integral for the separate cases $t+\nu<\eta$ and $t+\nu\geq\eta$. In this form, Equation (\ref{mom0}) can be rewritten with the use of the following lemma:
\begin{lemma}
Let $k\in\mathbb{N}$ and $F(x)$ be a positive real-valued integrable function. Assume that $\int_0^\infty x^{k+j+1} F(x)\text{d}x < \infty $, then
\[\resizebox{\linewidth}{!}{$\displaystyle\int_0^\infty \int_0^\infty x^j y^k F(x+y)\text{ d}x\text{d}y=\frac{k!j!}{(k+j+1)!}\int_0^\infty z^{k+j+1}F(z)\text{ d}z.$}\]
\end{lemma}
\begin{proof}
Write $z=x+y$, then
\[\begin{array}{rl}&\displaystyle 
\int_0^\infty \int_0^\infty x^j y^k F(x+y)\text{ d}x\text{ d}y\vspace{.1cm}\\
 \displaystyle=&\displaystyle \int_0^\infty \int_y^\infty (z-y)^j y^k F(z)\text{ d}z\text{ d}y\\
 \displaystyle=&\displaystyle \sum_{i=0}^j(-1)^i \binom{j}{i}\int_0^\infty \int_y^\infty z^{j-i}y^{k+i}F(z)\text{ d}z\text{ d}y\vspace{.1cm}\\
 \displaystyle=&\displaystyle \sum_{i=0}^j(-1)^i \binom{j}{i}\int_0^\infty z^{j-i}F(z) \int_0^z y^{k+i}\text{ d}y\text{ d}z\vspace{.1cm}\\\vspace{.1cm}
 \displaystyle=&\displaystyle\left(\sum_{i=0}^j(-1)^i \binom{j}{i}\frac{1}{1+i+k}\right)\int_0^\infty z^{1+k+j}F(z)\text{ d}z\\
 \displaystyle=&\displaystyle\frac{k!j!}{(k+j+1)!}\int_0^\infty z^{k+j+1}F(z)\text{ d}z.
\end{array}
\]
The last equality follows from a known identity involving the reciprocal of binomial coefficients (see \cite{trick16}, Corollary 2.2).
\end{proof}
Applying Lemma 3 to Equation (\ref{mom0}) gives
\begin{equation}\nonumber
\resizebox{1\linewidth}{!}{$
\begin{array}{c}\displaystyle
\frac{j!C_{(k,n)}}{(k+j-1)!}\int_0^\infty z^{k+j-1}\exp\left[-\frac{n\max[(z-\eta),0]^2}{2(1-\eta)}\right]\text{ d}z\vspace{0.1cm}\\
\displaystyle
=j!C_{(k,n)}\left(\frac{\eta^{k+j-1}}{(k+j-1)!}+\int_\eta^\infty \frac{z^{k+j-2}}{(k+j-2)!}\exp\left[-\frac{n(z-\eta)^2}{2(1-\eta)}\right]\text{ d}z\right).\\
\end{array}
$}
\end{equation}
Finally, using (\ref{norm}), we obtain the following insightful expression for the $j$th moment of an inter-transmission time
\begin{equation}\label{mom}
\mathbb{E}\left[\left(T^{(k,n)}\right)^j\right]\sim\mathbb{E}\left[\left(T_P^{(k,n)}\right)^j\right]= j!\frac{C_{(k,n)}}{C_{(k+j,n)}}.
\end{equation}
For the special case $\eta=0$ this reduces to
\begin{equation}\label{mom1}
\mathbb{E}\left[\left(T^{(k,n)}\right)^j\right]\sim\frac{j!}{(2n)^{\frac{j}{2}}}\frac{\Gamma\left[\frac{k}{2}\right]}{\Gamma\left[\frac{k+j}{2}\right]}.
\end{equation}
For the case $\eta>0$ we deduce
\begin{equation}\label{mom2}
\mathbb{E}\left[\left(T^{(k,n)}\right)^j\right]\rightarrow\frac{(k-1)!j!}{(k+j-1)!}\eta^j,\text{ as }n\rightarrow \infty.
\end{equation}
\subsubsection{Limiting Distributions}
The previous analysis allows us to determine the limiting distributions of $T^{(k,n)}$ as $n$ and $k$ grow large. We distinguish between two cases.
\\ \newline\noindent
\textbf{Case 1: $\mathbf{\boldsymbol{\eta}>0}$.}

First, Equation (\ref{mom2}) implies that for $\eta>0$ and $k\geq2$
\begin{equation}\label{lim}
\frac{1}{\eta}T^{(k,n)}\xrightarrow{d}\text{Beta}(1,k-1)\text{, as }n\rightarrow \infty,
\end{equation}
since all moments converge to the moments of the Beta distribution. For the density function of $T_P^{(k,n)}$ (and hence also that of $T^{(k,n)}$)  this implies that as $n\rightarrow \infty$,
\begin{equation}\label{lim2}
 f^{(k,n)}(t)\rightarrow
 \begin{cases}
 \frac{(k-1)}{\eta}\left(1-\frac{t}{\eta}\right)^{k-2}, &0\leq t\leq\eta,\\
 0, &\text{otherwise}.
 \end{cases}
\end{equation}
Furthermore, since $k\text{Beta}(1,k)\xrightarrow{d}\text{Exp}(1)$ for $k\rightarrow \infty$,
\begin{equation}\label{lim3}
\frac{k}{\eta}T^{(k,n)}\xrightarrow{d}\text{Exp}(1)\text{, as }n\rightarrow \infty\text{ and }k\rightarrow \infty.
\end{equation}
We can relate this result to Equation (\ref{sol2}). We know that the sum of $k$ consecutive inter-transmission times converges to $\eta$ as $n\rightarrow \infty$. Equation (\ref{sol2}) then tells us that if we look at an interval of length $\eta$, $k-1$ broadcasting times  are distributed uniformly in this interval, when $n  \rightarrow \infty$. Therefore, if we scale time by a factor $k$ we will see a Poisson process with intensity 1, as $k\rightarrow \infty$, and expect to see exponential inter-broadcasting times, which is in agreement with (\ref{lim3}).\\
\newline\noindent
\textbf{Case 2: $\mathbf{\boldsymbol{\eta}=0}$.}

For the case $\eta=0$, we can write the density function $f^{(k)}(t)$ of the limiting distribution of $\sqrt{\frac{n}{2}}T^{(k,n)}$ for $k\geq 2$ after scaling Equation (\ref{dens3}) as
\begin{align}
\displaystyle
f^{(k)}(t)&=\displaystyle\frac{4}{\Gamma\left[\frac{k-1}{2}\right]}\int_{0}^\infty(t+\nu)\nu^{k-2}\exp\left[-(t+\nu)^2\right]\text{d}\nu\vspace{.2cm}\nonumber\\\displaystyle
&=\displaystyle\frac{4}{\Gamma\left[\frac{k-1}{2}\right]}t\int_{0}^\infty\nu^{k-2}\exp\left[-(t+\nu)^2\right]\text{d}\nu\vspace{.2cm}\label{denslim}\\\displaystyle&\displaystyle-\frac{4}{\Gamma\left[\frac{k-1}{2}\right]}\int_{0}^\infty\nu^{k-1}\exp\left[-(t+\nu)^2\right]\text{d}\nu.\nonumber
\end{align}
Alternatively, performing integration by parts gives
\begin{equation}\label{denslim2}
f^{(k)}(t)=\frac{2(k-1)}{\Gamma\left[\frac{k-1}{2}\right]}\int_{0}^\infty\nu^{k-3}\exp\left[-(t+\nu)^2\right]\text{d}\nu.
\end{equation}
Combining (\ref{denslim}) and (\ref{denslim2}), we find after some rewriting
\begin{equation}\label{denslim3}
f^{(k)}(t)=\frac{k-2}{k-3}f^{(k-2)}(t)-\frac{\Gamma\left[\frac{k}{2}\right]}{\Gamma\left[\frac{k-1}{2}\right]}\frac{2t}{k-3}f^{(k-1)}(t).
\end{equation}
Direct computation gives
\[f^{(2)}(t)=\frac{2}{\sqrt{\pi}}e^{-t^2},\] 
and
\[f^{(3)}(t)=\sqrt{\pi}\text{erfc}(t).\]
All other distribution functions then follow from Equation (\ref{denslim3}).
Additionally, from Equation (\ref{mom1}) we have
\begin{equation}\label{mom4}
\mathbb{E}\left[\left(\sqrt{n}T^{(k,n)}\right)^j\right]\rightarrow\frac{j!}{2^{\frac{j}{2}}}\frac{\Gamma\left[\frac{k}{2}\right]}{\Gamma\left[\frac{k+j}{2}\right]}\text{, as }n\rightarrow \infty.
\end{equation}
Using the fact that
\[\lim_{k\rightarrow\infty}\frac{\Gamma\left[\frac{k}{2}\right]\left(\frac{k}{2}\right)^{\frac{j}{2}}}{\Gamma\left[\frac{k+j}{2}\right]}=1,\]
we find
\begin{equation}\label{mom5}
\mathbb{E}\left[\left(\sqrt{nk}T^{(k,n)}\right)^j\right]\rightarrow j!\text{, as }n\rightarrow \infty\text{ and }k\rightarrow \infty.
\end{equation}
This finally implies that
\begin{equation}\label{lim4}
\sqrt{nk}T^{(k,n)}\xrightarrow{d}\text{Exp}(1)\text{, as }n\rightarrow \infty\text{ and }k\rightarrow \infty,
\end{equation}
since all moments converge to the moments of the exponential distribution.

The last result can be related to Lemma 1. If we look at the special case $k=n$, we know all broadcasting attempts will be successful. Hence, the process of nodes broadcasting will be the same as the process of nodes attempting to broadcast, from which we know it converges to a Poisson process with rate 1, as $n\rightarrow \infty$, if we scale time by a factor $n$, because of Lemma 1. Furthermore, (\ref{lim4}) also tells us that for this case, if we scale time by a factor $n$, we expect to see exponential inter-transmission times with rate 1. So these results are in good agreement with each other.

\subsection{Discussion of Results}
We briefly discuss the implications of our results. First, Equation (\ref{expec1}) gives us insight in the effects of the short-listen problem. We found that without a listen-only period the number of transmissions per time unit scales as $\mathcal{O}(\sqrt{n})$, due to the short-listen problem. However, by introducing a listen-only period, the expected  number of transmissions per interval is bounded by $k/\eta$, see (\ref{expec2}). Hence, having $\eta=\frac{1}{2}$ as suggested in \cite{trick3} reduces the number of redundant transmissions and resolves the short-listen problem.

However, one could also consider using lower values of $\eta$. For example, one could consider setting $\eta=\frac{1}{4}$. We know this roughly doubles the expected number of redundant transmissions, but it could also double the propagation speed. Indeed, whenever a node gets updated, it only has to wait for a period of at least $\tau_l/4$ before propagating the update, as opposed to a period of $\tau_l/2$. Hence, the choice of $\eta$ involves a trade-off. Lowering $\eta$ increases the propagation speed, but also the number of transmissions.  In order to get a full understanding of this trade-off, more knowledge is needed on the propagation speed of the algorithm. Combined with results from this paper, this knowledge could provide guidelines for how to optimally set the Trickle parameters and the length of the listen-only period.

Second, Equations (\ref{densk1}), (\ref{dens}), (\ref{lim2}), (\ref{lim3}), (\ref{denslim3}) and (\ref{lim4}) provide insight in the distributions of inter-transmission times. In our derivations, we have assumed that transmissions are instantaneous. However, in reality transmissions take some time, although this time is generally short compared to the interval length $\tau_h$. Using the inter-transmission time distributions, one can estimate the probability that transmissions would overlap in time and hence interfere. Potentially, one can use this knowledge to optimize the medium access protocol and the capacity of the wireless network.
\subsection{Simulations}
The model we have developed assumes that the process of nodes attempting to broadcast is Poisson with rate $n$. Lemma 1 shows that this is true as $n$ grows large. We now investigate how well this assumption holds in networks with only a few nodes. In order to do so, we compare our analytic results for the message count and the distribution of inter-transmissions times with simulations done in Mathematica. We simulate a lossless single-cell network using the Trickle algorithm, where transmissions occur instantaneously. 
\begin{figure}[h!]
\centering
\includegraphics[width=8.4cm]{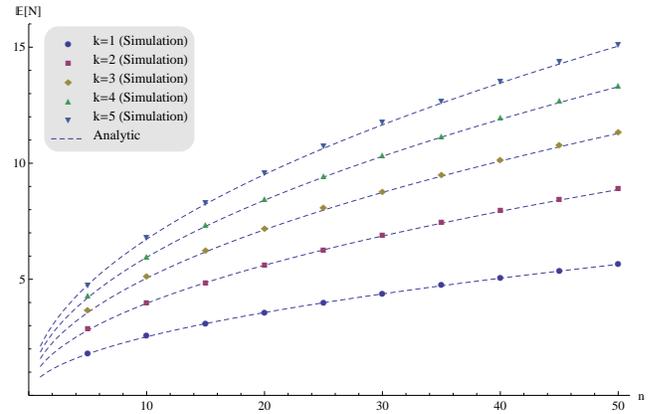}
\caption{Mean number of transmissions per interval for \boldmath$\eta=0$: analysis vs simulation.\vspace{.95cm}}\label{fig3}
\end{figure}
\begin{figure}[h!]
\centering
\includegraphics[width=8.4cm]{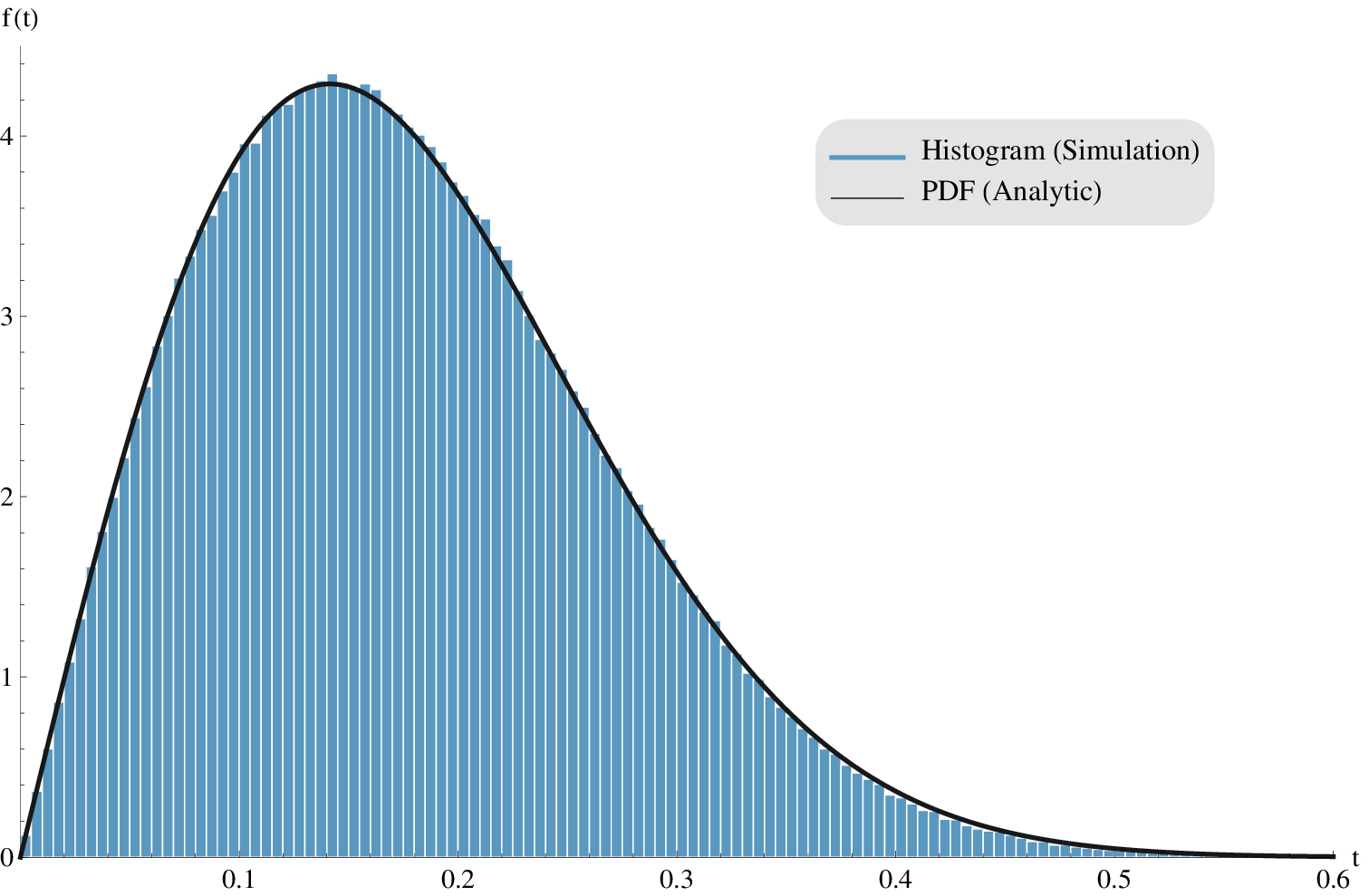}
\caption{Density plot of inter-transmission times for \boldmath$\eta=0$, \boldmath$k=1$ and \boldmath$n=50$: analysis vs simulation.\vspace{.95cm}}\label{fig4}
\end{figure}
\begin{figure}[h!]
\centering
\includegraphics[width=8.4cm]{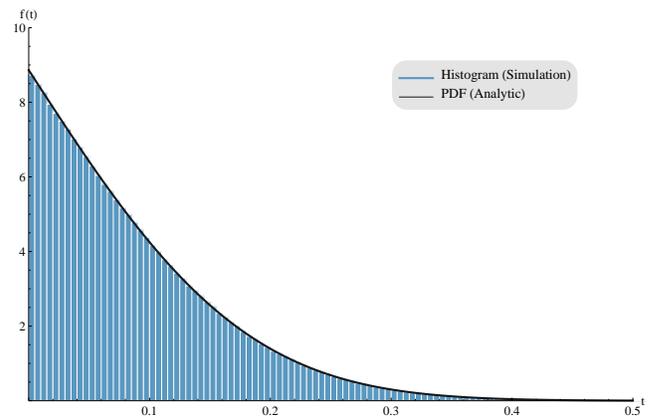}
\caption{Density plot of inter-transmission times for \boldmath$\eta=0$, \boldmath$k=3$ and \boldmath$n=50$: analysis vs simulation.}\label{fig5}
\end{figure}
In Figure \ref{fig3}, we compare simulation results  for the case $\eta=0$ with the analytic results from Equation (\ref{expec1}). For each combination of $k$ and $n$, we simulate a network consisting of $n$ nodes using the Trickle algorithm and average over $10^3$ runs of 100 virtual time units. Each run, we use a different interval skew chosen uniformly at random for each of the nodes. We see that the simulation results and the analytic results coincide very well, even for small cell-sizes, and that the analytic results provide a conservative estimate for the mean number of transmissions per interval. 

In Figure \ref{fig4}, we compare simulation results for the case $\eta=0$, $k=1$ and $n=50$ with the analytic results from Equation (\ref{densk1}). We show a histogram of the observed inter-transmission times obtained from $10^3$ runs of 100 virtual time units and the probability density function from (\ref{densk1}). We see a very good match between the analytic result and simulations, even though we consider a network consisting of only 50 nodes.

In Figure \ref{fig5}, we compare (\ref{dens}) with simulations for the case $\eta=0$, $k=3$ and $n=50$. Again, both results are in good agreement with each other, despite the relatively small size of the network. Note, also, that the asymptotically exponential behavior from (\ref{lim4}) can already be seen in the density plot.

In Figure \ref{fig6} we compare simulation results for the case $\eta=\frac{1}{2}$ with analytic results from Equation (\ref{expec2}). Here, also, we find that the analysis gives an accurate estimate for the mean number of transmissions in small networks. From our analysis we also know that $\mathbb{E}[N^{(k,n)}]$ should converge to $2k$ as $n$ grows large. From the graph we see that this convergence is quite slow.

In Figures \ref{fig7} and \ref{fig8} we compare simulation results for $k=1$ and $k=3$ respectively with the analytic results from Equations (\ref{densk1}) and (\ref{dens}), where $\eta=\frac{1}{2}$ and $n=50$. Like before, we see a very good match between the analytic results and simulations. Note, that the asymptotic behavior from (\ref{lim}) can already be recognized in the density plot of Figure \ref{fig8}. 

Additionally, we have simulated more realistic network settings in the OMNeT++ simulation tool. As mentioned in Section 3.3, in reality, transmissions do not occur instantaneously but take some time $M$, hence collisions can occur. To investigate the effect of not having instantaneous and possibly colliding transmissions, we have simulated the same scenarios as we did with Mathematica and compared the results. The simulator incorporates IEEE standard 802.15.4. 

We find that the OMNeT++ and Mathematica simulations produce nearly indistinguishable results, which is not surprising. Since $\tau_h$ tends to be very large compared to a transmission length $M$, the probability of nodes trying to broadcast simultaneously tends to be very small. Furthermore, since we are dealing with a single-cell network, the CSMA/CA protocol prevents all collisions in the OMNeT++ simulations. Therefore, both simulators provide nearly the same results. Since the Mathematica simulations are computationally less demanding and much more easily reproducible, we have omitted the OMNeT++ simulation data.
\section{Multi-Cell Network}
Suppose now that we have a network consisting of $n^2$ nodes placed on a square grid, where not all nodes are able to directly communicate with each other. Instead, each node has a fixed transmission range $R$, which means that when a node sends a message, only nodes within a distance $R$ of the broadcaster receive the message. While a full-fledged analysis of multi-cell networks is beyond the scope of the present paper, we now briefly examine how the results for the single-cell scenario obtained in Subsection 3.2 can be leveraged to derive a useful approximation. We denote by $N_{MC}^{(k,n,R)}$ the number of transmissions during an interval of length $\tau_h$ for a given threshold value $k$ in a cell consisting of $n\times n$ nodes with broadcasting range $R$. Hence, we are interested in determining the behavior of $\mathbb{E}\left[N_{MC}^{(k,n,R)}\right]$.
\begin{figure}[h!]
\centering
\includegraphics[width=8.4cm]{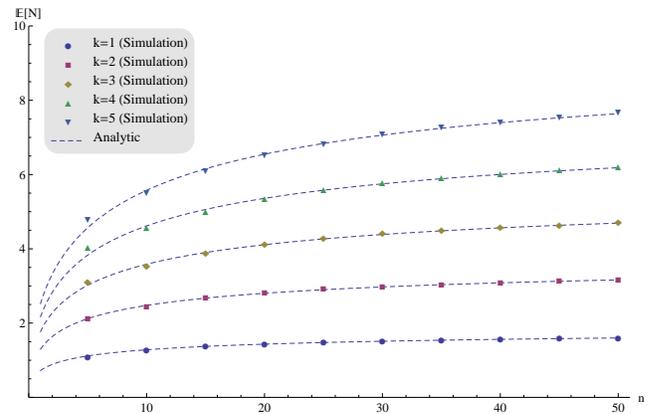}
\caption{Mean number of transmissions per interval for \boldmath$\eta=\frac{1}{2}$: analysis vs simulation.\vspace{.9cm}}\label{fig6}
\end{figure}
\begin{figure}[h!]
\centering
\includegraphics[width=8.4cm]{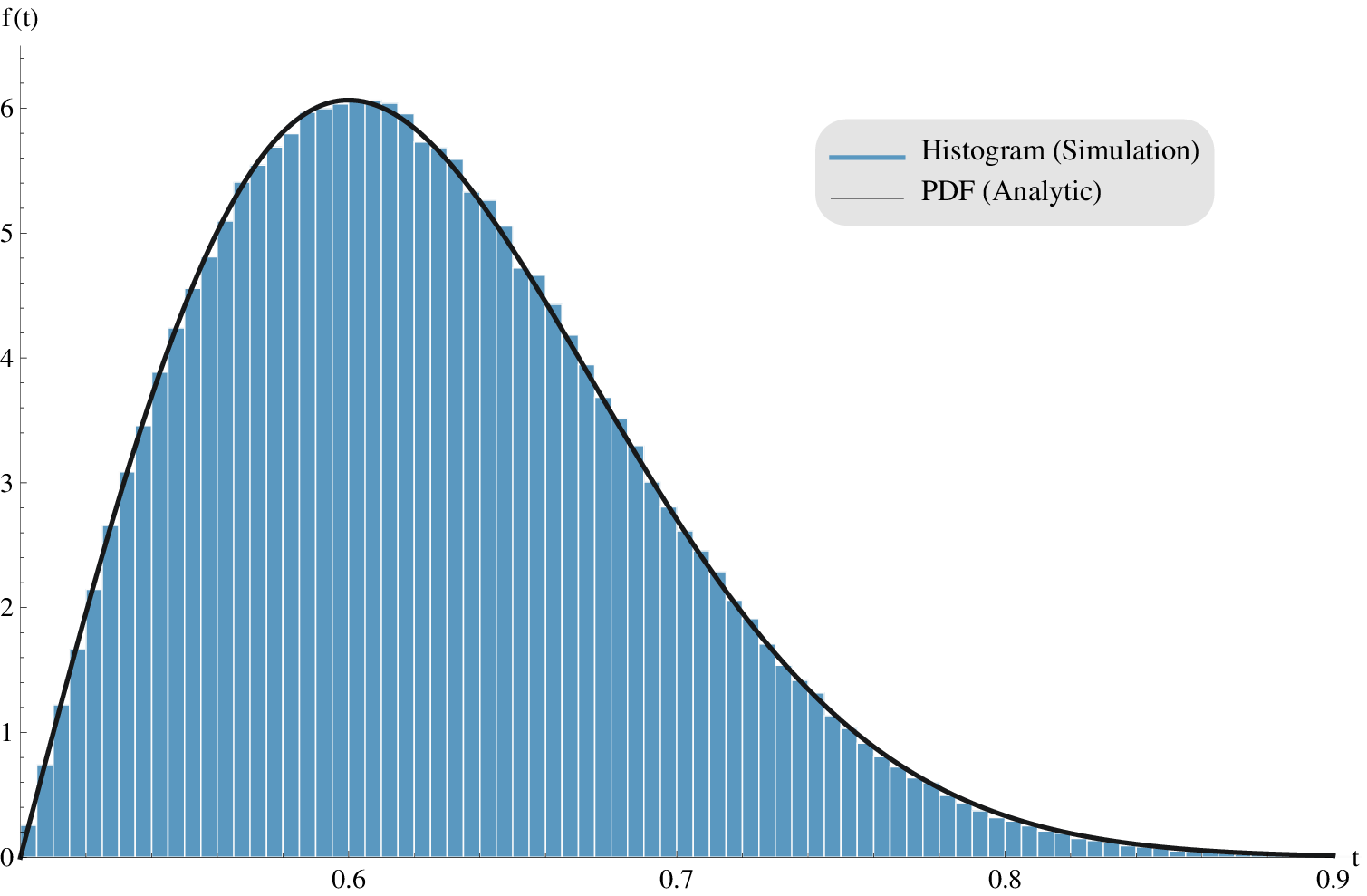}
\caption{Density plot of inter-transmission times for \boldmath$\eta=\frac{1}{2}$, \boldmath$k=1$ and \boldmath$n=50$: analysis vs simulation.\vspace{.9cm}}\label{fig7}
\end{figure}
\begin{figure}[h!]
\centering
\includegraphics[width=8.4cm]{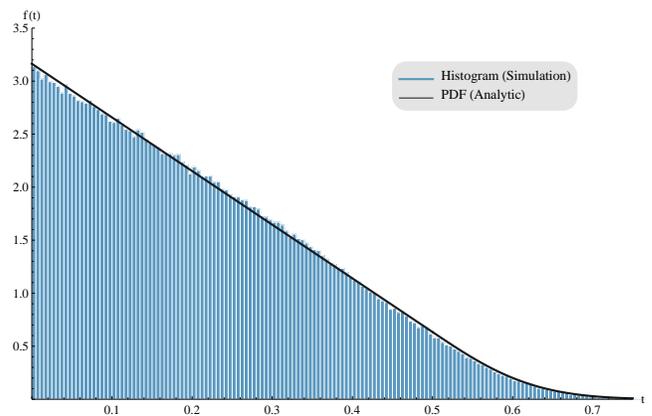}
\caption{Density plot of inter-transmission times for \boldmath$\eta=\frac{1}{2}$, \boldmath$k=3$ and \boldmath$n=50$: analysis vs simulation.}\label{fig8}
\end{figure}
\vfill\eject
\subsection{Approximation}
We will use the analytical results from Section 3 to develop an approximation for the expected number of transmissions per interval in multi-cell networks. Let us denote by  $S(R)$ the number of nodes a single broadcasting node  having a broadcasting range of $R$ can reach, i.e., the size of a single broadcasting-cell. Heuristically, we can reason as follows. We have an $n\times n$ grid consisting of approximately $n^2/S(R)$ distinct non-overlapping broadcasting-cells. Assuming each cell behaves independently of the others we can approximate the expected number of transmissions per interval in the multi-cell case as follows:
\begin{equation}\label{approx}
\resizebox{.88\linewidth}{!}{$\displaystyle
\mathbb{E}\left[N^{(k,n,R)}_{MC}\right] \approx \frac{n^2}{S(R)}\mathbb{E}\left[N^{(k,S(R))}\right]\approx \frac{n^2}{S(R)}\frac{C_{(k+1,S(R))}}{C_{(k,S(R))}}.
$}
\end{equation}

Using the fact that $S(R)\sim \pi R^2$ and Equations (\ref{expec1}) and (\ref{expec2}), we can get more insight in the behavior of our approximation. For $\eta=0$ we have that as $R$ grows large:
\[\frac{n^2}{S(R)}\mathbb{E}\left[N^{(k,S(R))}\right]\approx \sqrt{\frac{2}{\pi}}\frac{n^2}{R}\frac{\Gamma\left[\frac{k+1}{2}\right]}{\Gamma\left[\frac{k}{2}\right]}.\]
Similarly, when a listen-only period is used, i.e., $\eta>0$, we have that as $R$ grows large:
\[\frac{n^2}{S(R)}\mathbb{E}\left[N^{(k,S(R))}\right]\approx \frac{n^2}{R^2}\frac{k}{\pi \eta}.\]
Consequently, we see how the short-listen problem potentially plays a role in multi-cell networks.
\subsection{Simulations}
We compare the approximation with simulations in order to evaluate its accuracy. We do so by plotting the ratio
\begin{equation}\label{ratio}
\theta(k,n,R)=\mathbb{E}\left[N^{(k,n,R)}_{MC}\right]/\left(\frac{n^2}{S(R)}\frac{C_{(k+1,S(R))}}{C_{(k,S(R))}}\right).
\end{equation}
In Mathematica we simulate a network consisting of 50$\times$50 nodes placed on a grid for several values of $k$ and $R$. For each combination of $k$ and $R$, we simulate a lossless multi-cell network for 100 virtual time units. Each run, we use a different randomly chosen interval skew for the nodes. We use toroidal distance in order to cope with edge effects. 

In Figure \ref{fig9} we show a plot of (\ref{ratio}) for the case $\eta=0$. We see that the estimate given in Equation (\ref{approx}) tends to slightly underestimate the number of transmissions per interval. However, the approximation is accurate within a factor 1.2. Furthermore, as $k$ increases, the approximation becomes more accurate.  We also note that for fixed $k$ the accuracy remains fairly constant as $R$ grows.

Finally, let us consider the effect of a listen-only period. In Figure \ref{fig10} we show a plot of (\ref{ratio}) for the case $\eta=\frac{1}{2}$.
\begin{figure}[h!]
\centering
\includegraphics[width=8.4cm]{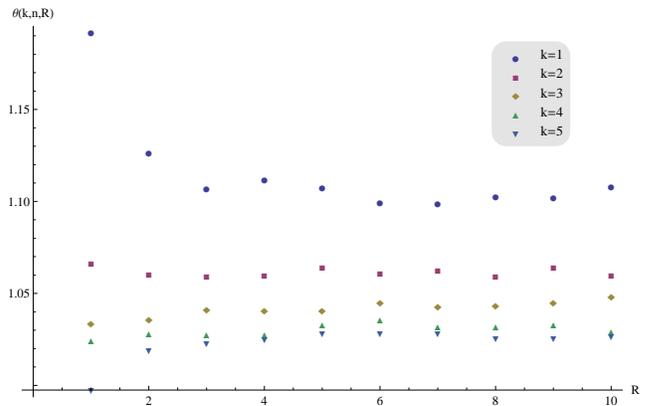}
\caption{Ratio of simulated transmission count and approximation for \boldmath$\eta=0$.}\label{fig9}
\end{figure}
\begin{figure}
\centering
\includegraphics[width=8.4cm]{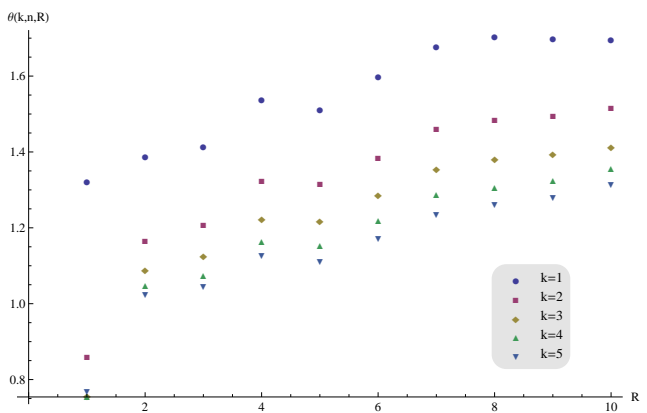}
\caption{Ratio of simulated transmission count and approximation for \boldmath$\eta=\frac{1}{2}$.}\label{fig10}
\end{figure}
We see again that the estimate given in Equation (\ref{approx}) tends to underestimate the number of transmissions per interval, more than for the case $\eta=0$. The error slowly increases with the transmission range $R$. However, as $k$ increases, the approximation becomes more accurate.

Note that in \cite{trick6} a method is presented for estimating the message count in synchronized networks. It is shown that this method also gives an accurate approximation for unsynchronized networks with $\eta=\frac{1}{2}$. However, the method is not suitable for smaller values of $\eta$ and assumes a uniformly random spatial distribution of the nodes. Hence, if one is interested in listen-only periods of different length or regular network topologies, our approximation is preferable.

\section{Conclusion}
In this paper, we presented a generalized version of the Trickle algorithm with a new parameter $\eta$, which allows us to set the length of a listen-only period. This parameter can greatly increase the speed at which the algorithm can propagate updates, while still controlling the number of transmissions. Furthermore, we have shown that this parameter influences how the transmission load is distributed among nodes in a network. We then presented a mathematical model describing how the message count and inter-transmission times of the Trickle algorithm depend on its various parameters. We were particularly interested in the effect of a listen-only period on the message count and on the inter-transmission time distributions. 

For single-cell networks we derived the distribution of inter-transmission times and the joint distribution of consecutive inter-transmission times. We showed how they depend on the redundancy constant $k$, the network-size $n$ and the length of a listen-only period $\eta$. We also investigated their asymptotic behavior as $n$ and $k$ go to infinity. These distributions give insight in the energy-efficiency of the algorithm and the probability that nodes try to broadcast simultaneously. These insights contribute to optimizing the design and usage of the Trickle algorithm in WSNs.

Specifically, we showed that in a network without a listen-only period the expected number of transmissions grows as $\mathcal{O}(\sqrt{n})$, proving a conjecture from \cite{trick3}, and we identified that the pre-factor is $\sqrt{2}\Gamma\left[\frac{k+1}{2}\right]/\Gamma\left[\frac{k}{2}\right]$. Additionally, we showed that, when a listen-only period is used, the number of transmissions per interval is bounded from above by $\frac{k}{\eta}$, proving a second conjecture from \cite{trick3}.

We have also performed a simulation study in Mathematica and the OMNeT++ simulation tool. We compared our analytic results, which hold for very large networks with instantaneous transmissions, with simulation results of small and more realistic wireless networks. We found a very good match between the analytic results and the simulations.

Additionally, we used the results from the single-cell analysis to get an approximation for the message count in multi-cell networks. These results were also compared to simulation results from Mathematica and the OMNeT++ simulation tool. The approximation proved to be fairly accurate, in particular for small values of $\eta$. A more comprehensive investigation of multi-cell networks would be an interesting topic for further research.

Finally, we note that the speed at which the Trickle algorithm can disseminate updates in WSNs is a topic of ongoing research for the authors. Combined with results from this paper, this could provide insight in how to optimally set the Trickle parameters and the length of the listen-only period. 
\section{Acknowledgments}
The authors would like to thank Guido Janssen for useful discussions regarding the integral equations presented in this paper. Additionally, they are grateful to Marc Aoun for his contributions to the performed OMNeT++ simulations. The research of the third author was done in the framework of the IAP BESTCOM program, funded by the Belgian government. 
%ACKNOWLEDGMENTS are optional
%
% The following two commands are all you need in the
% initial runs of your .tex file to
% produce the bibliography for the citations in your paper.
\bibliographystyle{abbrv}
\bibliography{sigproc}  % sigproc.bib is the name of the Bibliography in this case
% You must have a proper ".bib" file
%  and remember to run:
% latex bibtex latex latex
% to resolve all references
%
% ACM needs 'a single self-contained file'!
%
%APPENDICES are optional
\vfill\eject
\appendix
%Appendix A
\section{Proof of Lemma 1}
%\balancecolumns
In order to prove Lemma 1 we will use the following theorem (see \cite{trick8}, Proposition 11.2.VI).
\begin{theorem}
Let $M$ be a simple stationary point process on $X=\mathbb{R}$ with finite intensity $\lambda$, and let $M_n$ denote the point process obtained by
superposing $n$ independent replicates of $M$ and dilating the scale of $X$ by a factor $n$. Then as $n\rightarrow \infty$, $M_n$ converges weakly to a Poisson process with parameter measure $\lambda \mu_L(\cdot)$, where $\mu_L(\cdot)$ denotes the Lebesgue measure on $\mathbb{R}$.
\end{theorem}
Let $N$ be the process of an arbitrary node's broadcasting moments. Then evidently the superposition of $n$ of these processes $N_n$ gives us the process of broadcasting attempts by all the nodes in our cell. Hence if we show that $N$ is a simple stationary point process with finite intensity $\lambda=1$, we can apply Theorem 2, resulting in Lemma 1.

Let $\{B_i\}$ with $i\in \mathbb{Z}$ be the sequence of broadcasting times for an arbitrary node defining our process $N$. Then if we denote by $S$ the interval skew and assume it is uniform, i.e., $S\sim U[0,1]$, we can write
\[B_i \sim (S+i)+T_i\text{, }i\in \mathbb{Z}.\]
\vfill\eject
Here $T_i$ represents the timer of the corresponding node in the $i$th interval, thus $T_i\sim U[\eta,1]$. It is easy to see that $N$ is a simple point process with intensity $\lambda=1$.

It remains to show stationarity. Let $N^s$ be the point process resulting from applying a shift operator to the process $N$. That is, if  $\{B^s_i\}$ with $i\in \mathbb{Z}$ is the sequence of events for the shifted process $N^s$ we have
\[B^s_i \sim (S+i)+T_i+s\text{, }i\in \mathbb{Z}.\]
We show that $\{B_i\}$ and $\{B_i^s\}$ with $i\in \mathbb{Z}$ have the same distribution, which means that the processes $N$ and $N^s$ coincide in distribution, implying stationarity. We first write\vspace{.1cm}
\[ \{B_i^s\}_{i\in \mathbb{Z}}=\{(S+i)+T_i+s\}_{i\in \mathbb{Z}}\]\[=\{(S+s-\lfloor S+s\rfloor)+(i+\lfloor S+s\rfloor+T_i)\}_{i\in \mathbb{Z}}\]
Now note that $\left(S+s-\lfloor S+s\rfloor\right)\sim U[0,1] \sim S$. Furthermore, $(i+\lfloor S+s\rfloor+T_i) \sim \left(i+\lfloor S+s\rfloor +T_{i+\lfloor S+s\rfloor}\right)$, since the $T_i$ are independent and identically distributed. Moreover, $\left\{i+\lfloor S+s\rfloor +T_{i+\lfloor S+s\rfloor}\right\}_{i\in \mathbb{Z}}\sim \{i+T_{i}\}_{i\in \mathbb{Z}}$. This allows us to conclude
\[\left\{B_i^s\right\}_{i\in \mathbb{Z}}\sim\{B_i\}_{i\in \mathbb{Z}}.\]
Consequently Theorem 2 applies and Lemma 1 follows. 
\end{document}